\documentclass[12pt]{article}
\usepackage{graphicx}
\usepackage{color}
\usepackage{amsmath}
\usepackage{amssymb}
\usepackage{amscd}
\usepackage{amsthm}
\usepackage{amsopn}
\usepackage{xspace}
\usepackage{verbatim}
\usepackage{amsmath}
\usepackage{amsfonts}
\usepackage{epic}
\usepackage{eepic}
\usepackage{stmaryrd}
\usepackage{empheq}
\usepackage{multirow}
\usepackage{colortbl}
\usepackage[active]{srcltx}
\usepackage[a4paper,margin=3cm]{geometry}
\definecolor{red}{rgb}{1,0,0}
\definecolor{green}{rgb}{0,1,0}
\definecolor{SeaGreen}{RGB}{46,139,87}
\definecolor{Maroon}{RGB}{128,0,0}

\def\qed{\hbox {\hskip 1pt \vrule width 4pt height 6pt depth 1.5pt
        \hskip 1pt}\\}

\newcommand{\Supp}{\textrm{Supp~}}
\newcommand{\N}{\mathbb{N}}

\newcommand{\C}{{\mathbb{C}}}
\newcommand{\R}{{\mathbb{R}}}

\newcommand{\A}{{\mathcal A}}

\def\Jg {{\mathcal J}}

\def\Eg {{\mathcal E}}
\newcommand{\LL}{\mathcal L}

\def\Sg {{\mathcal S}}

\renewcommand {\Re}{{\rm Re\,}}

\newcommand {\pa}{\partial}

\newcommand {\ar}{\to}

\def\Ai{\text{\rm Ai\,}}

\def\0{\mathbf  0}

\def\XXint#1#2#3{{\setbox0=\hbox{$#1{#2#3}{\int}$ }
\vcenter{\hbox{$#2#3$ }}\kern-.6\wd0}}

\newcommand{\Union}{\mathop{\bigcup}\limits}

 

\errorcontextlines=0 \numberwithin{equation}{section}
\theoremstyle{plain}
\newtheorem{theorem}{Theorem}[section]
\newtheorem{lemma}[theorem]{Lemma}

\newtheorem{proposition}[theorem]{Proposition}
\newtheorem{assumption}[theorem]{Assumption}
\newtheorem{definition}[theorem]{Definition}
\newtheorem{remark}[theorem]{Remark}

\newcommand {\jf}{{\bf j}}

\title{Spectral semi-classical analysis of a complex Schr\"odinger operator in  exterior domains} 

\author{ Y. Almog, Department of Mathematics, Louisiana State University,\\
    Baton Rouge, LA 70803, USA,\\~\\
 D. S. Grebenkov, 
 Laboratoire de Physique de la Mati\`ere Condens\'ee, \\ 
CNRS and Ecole Polytechnique, F-91128 Palaiseau, France,
\\
  and \\
  B. Helffer, Laboratoire de Math\'ematiques Jean Leray, \\CNRS and Universit\'e de Nantes, \\
  2 rue de la Houssini\`ere, 44322 Nantes Cedex France.}

\date{}

\begin{document}
\bibliographystyle{siam}

\maketitle
\begin{abstract}
  Generalizing previous results obtained for the spectrum of the
  Dirichlet and Neumann realizations in a bounded domain of a
  Schr\"odinger operator with a purely imaginary potential $-h^2\Delta+iV$
  in the semiclassical limit \break $h\to0$ we address the same problem
  in exterior domains. In particular we obtain the left margin of the
  spectrum, and the emptiness of the essential part of the spectrum
  under some additional assumptions.
\end{abstract}

\section{Introduction}
\label{sec:1}

Let $\Omega = K^c$, where $K$ is a compact set with smooth boundary in
$\mathbb R^d$ with $d\geq 1$. Consider the operator
\begin{subequations}
  \label{eq:1D}
  \begin{equation}
\A_h^D = -h^2\Delta + i\, V \,,
\end{equation}
defined on
\begin{equation}
\label{eq:2D}
  D(\A_h^D)= \{ u \in H^2 (\Omega) \cap H_0^1(\Omega) \,,\, V u \in L^2(\Omega)\}\,,
\end{equation}
\end{subequations} 
or
\begin{subequations}
  \label{eq:1N}
  \begin{equation}
\A_h^N = -h^2\Delta + i\, V \,,
\end{equation}
defined on
\begin{equation}
\label{eq:2N}
  D(\A_h^N)= \{ u \in H^2 (\Omega) \,,\, V u \in L^2(\Omega)\,,\,  \partial_\nu u = 0 \mbox{ on } \partial \Omega\} \,,
\end{equation}
\end{subequations}
where $V$ is a $C^\infty$-potential in $\overline{\Omega}$, $\nu$ is
pointing outwards of $\Omega$.\\
The quadratic forms respectively read
\begin{equation}
\label{eq:formR}
u \mapsto  q_V^D(u) := h^2 \,\| \nabla u\|^2_{ \Omega} 
  + i \int_{ \Omega} V(x)  | u(x) |^2 \, dx \,,
\end{equation}
where the form domain is
\begin{equation*}
\mathcal V^D (\Omega) =\{ u\in H^1_0(\Omega)\,,\, | V|^\frac 12 u \in L^2(\Omega)\}\,,
\end{equation*}
and 
\begin{equation}
\label{eq:formRN}
u \mapsto  q_V^N(u) := h^2 \,\| \nabla u\|^2_{ \Omega} 
  + i \int_{ \Omega} V(x)  | u(x) |^2 \, dx \,,
\end{equation}
where the form domain is
\begin{equation*}
\mathcal V^N (\Omega) =\{ u\in H^1(\Omega)\,,\, | V|^\frac 12 u \in L^2(\Omega)\}\,.
\end{equation*}
Although the forms are not necessarily coercive when $V$ changes 
sign, a natural definition, via an extended Lax-Milgram theorem, can
be given for $\mathcal A_h^D$ or $\mathcal A_h^N$ under the condition
that there exists $C>0$ such that
\begin{equation}\label{ass0}
|\nabla V(x)| \leq C \sqrt{V(x)^2 +1}\,,\quad  \forall x \in \overline{\Omega}\,.
\end{equation}
We refer to \cite{AlHel, GHH,AGH} for this point and the
characterization of the domain of $\mathcal A_h^\#$, where the
notation $\#$ is used for $D$ (Dirichlet) or $N$ (Neumann). 

Note that (\ref{ass0}) is satisfied for $V(x) =x_1$
(Bloch-Torrey equation) which is our main motivating example (see
\cite{Grebenkov17}).\\
In this last case and when $K = \emptyset$, it has been demonstrated that the
spectrum is empty \cite{Alm,Hen}.  Our aim is now to analyze, when $K$
is not empty, the two following properties
\begin{itemize}
\item 
the emptiness of the essential spectrum, although the resolvent is not
compact when $d\geq 2\,$,
\item 
the non emptiness of the spectrum and the extension of the
semi-classical result of Almog-Grebenkov-Helffer \cite{AGH} concerning
the bottom of the real part of the spectrum.
\end{itemize}
Since the case $d=1$ was analyzed in \cite{GHH,Hen}, we will assume
from now on that
\begin{equation}
d\geq 2\,.
\end{equation}

The study of the spectrum of the operator \eqref{eq:1D} in bounded
domains began in \cite{Alm} where a lower bound on the left margin of
the spectrum has been obtained. In \cite{Hen} the same lower bound has
been obtained using a different technique allowing for resolvent
estimates (and consequently semigroup estimates), that are not
available in \cite{Alm}. In \cite{AlHen} an upper bound for the
spectrum has been obtained under some rather restrictive assumptions
on $V$. In \cite{AGH} these assumptions were removed and an upper
bound (and a lower bound) for the left margin of the spectrum has been
obtained not only for \eqref{eq:1D} but also for \eqref{eq:1N} as well
as for the Robin realization and for the transmission problem,
continuing and relying on some one-dimensional result obtained in
\cite{GHH} and on the formal derivation of the relevant quasimodes
obtained in \cite{GH}. 

The rest of this contribution is arranged as follows: in the next
section we list our main results. Section 3 is devoted to the
emptiness of the essential spectrum of \eqref{eq:1N} under some
conditions on the potential. In some case we confine the essential
spectrum to a certain part of the complex plain whereas in other
cases we show that it is empty. The methods in Section 3 are equally
applicable to \eqref{eq:1D} as well as to the Robin realization and to
the transmission problem. In Section 4 we derive the left margin of
the spectrum in the semi-classical limit, by using the same method as
in \cite{AGH}. In Section 5 we present some numerical results, and in
the last section we emphasize some points which were not addressed
within the analysis.

\section{Main results} \label{s5}

\subsection{Analysis of the essential spectrum.}

It is clear that there is no essential spectrum when $V \to
+\infty$ as \break $|x|\to +\infty$ but we are motivated by
the typical example $V(x)= x_1$ with $d\geq 2\,$, in which case $V$ can
tend to $-\infty$ as well.  To treat this example, we need, in addition t
\eqref{ass0}, the following assumption:

\begin{assumption}\label{ass7} 
There exists $R>0$ such that $K \subset B(0,R)$ and a potential $V_0$
satisfying
\begin{enumerate}
\item 
There exists $C>0$ such that, for $\forall x \in \mathbb R^d\,$, 
\begin{equation}\label{ass3}
\sum_{|\alpha|=2} | \partial_x^\alpha V_0 (x)| \leq C\,  |\nabla V_0|^{2/3}\,,
\end{equation}
\item
There exists $c>0$ such that, for $\forall x \in \mathbb R^d\,$, 
\begin{equation}\label{ass4}
0 < c \leq |\nabla V_0 (x)|\,,
\end{equation}
\end{enumerate}
and such that $V= V_0$ outside of $B(0,R)$.
\end{assumption}
A necessary condition for $V$ at the boundary of $B(0,R)$ is that
$\partial_\nu V =0$ at some point of the boundary.  If not, any $C^1$
extension $V_0$ of $V$ inside $B(0,R)$ has a critical point in
$B(0,R)$.\\
Under Assumption \ref{ass7}, we have:
\begin{theorem}\label{theorem2.3} 
For $\#\in \{D,N\}$, and under Assumption \ref{ass7}, for any $\Lambda
\in \mathbb R$, there exists $h_0 >0$ such that for $h\in (0,h_0]$ the
operator $\mathcal A_h^\#$ has no essential spectrum in $\{z\in
\mathbb C\,\,|\, \Re z \leq \Lambda h^\frac 23\}$.
\end{theorem}
We now introduce the stronger condition (which is 
Assumption~\ref{ass7} with $V_0 = \jf x_1$):
 \begin{assumption}\label{ass5}
The potential $V$ is given in $\R^d\setminus K$ by 
\begin{equation*}
V(x)= \jf\,  x_1+\tilde{V}\,,
\end{equation*}
where $\tilde{V}\in C^1(\R^d)$ and satisfies $\tilde{V}\to0$ as
  $|x|\to +\infty$.  
\end{assumption}

\begin{theorem} \label{theorem2.5} 
Under Assumption \ref{ass5}, the operator $\mathcal A_h^\#$ has no
essential spectrum.
\end{theorem}
  
In other words the spectrum of the operator $\mathcal A_h^\#$ is
either empty, or discrete.  This spectral property of the operator
$\mathcal A_h^\#$ contrasts with a continuous spectrum of the Laplace
operator in the exterior of a compact set.  Adding a purely imaginary
potential $V$ to the Laplace operator drastically changes its spectral
properties.  As a consequence, the limiting behavior of the operator
$-\Delta + ig V$ as $g\to 0$ is singular and violates conventional
perturbation approaches that are commonly used in physical literature
to deal with this problem (see discussion in \cite{Grebenkov17}).
This finding has thus important consequences for the theory of
diffusion nuclear magnetic resonance (NMR).  In particular, the
currently accepted perturbative analysis paradigm has to be
fundamentally revised.
%{\clm For instance, the commonly used Gaussian phase
%approximation for the macroscopic NMR signal in the exterior of a
%compact set (e.g., in the extracellular space) is not applicable at
%high $g$.}\\
%{\clg I don't understand the last statement. I hope that Denis will be
%  ready to clarify what he means by ``Gaussian phase
%approximation'', and why has small $g$ changed into large (high)?}

\begin{remark}
  It is not clear at all whether the spectrum of $-\Delta+i\, \jf x_1 $
  remains empty if we add to it a potential $V$ such that $(-\Delta + i\jf
  x_1)^{-1} V$ is compact (for example $V$ with compact support).  In
  fact, one may construct  a real valued $V\in C^1$ with compact
  support, such that $\sigma(-d^2/dx^2+i(x+V))\neq\emptyset$. However, if we consider the
  operator $\A=-\Delta+ix_1+iV(x^\prime)$ acting on $\R^d$, where $x^\prime\in\R^{d-1}$ so that
  $x=(x_1,x^\prime)$. Since $\A$ is separable in $x_1$ and $x^\prime$ we may write
  \begin{displaymath}
    e^{-t\A}=e^{-t(-\partial^2{x_1}+ix_1)}\otimes e^{-t(-\Delta_{x^\prime}+iV(x^\prime))}\,,
  \end{displaymath}
Consequently (see also \cite[Section 4]{AGH}) we have
\begin{displaymath}
  \|e^{-t\A}\|\leq Ce^{-t^3/12}\,.
\end{displaymath}
It follows that $\sigma(\A)=\emptyset$. If consider the Dirichlet or Neumann
realization of $\A$ in $\Omega$, then we may use we may use the same
procedure detailed in the proof of Theorem \ref{theorem2.5} to
conclude that $\sigma_{ess}(\A^\#)=\emptyset$.
\end{remark}
I have dropped the appendix, but kept the statement that it can be
done. We can go one step further and drop the entire statement, but I
suggest that we still keep something in that spirit.

\begin{remark}
  Let
  \begin{displaymath}
    V= a \, x^2_1 + \tilde{V} \,,
  \end{displaymath}
where $\tilde{V}\in C^1$ satisfies
$\tilde{V}\xrightarrow[|x|\to + \infty]{}0$ and $a > 0$. \\
Then (with $h=1$)
\begin{displaymath}
  \sigma_{ess}(\A_1^\#)= \Union_{
    \begin{subarray}{c}
      r\geq0 \\
      n\in\N
    \end{subarray}}\{ e^{i\pi/4}a^{1/2}(2n-1)+r \}\,.
\end{displaymath}
\end{remark}

The proof is very similar to the proof of Theorem \ref{theorem2.5} and
is therefore skipped. Note that, in the limit $a\to0^+$,
$\sigma_{ess}(\A_1^\#)$ tends to the sector $0\leq\arg z\leq \pi/4$. This is,
once again, not in accordance with the guess that the essential
spectrum tends to $\R_+=\sigma_{ess}(-\Delta)$.

\subsection{Semi-classical analysis of the bottom of the spectrum.} 

We begin by recalling the assumptions made
 in \cite{Hen, AlHen,AGH} (sometimes in
a stronger form) while obtaining a bound on the left margin of the
spectrum of $\A^\#_h$ in a bounded domain.  In contrast, we
consider here a domain which lie in the exterior of a bounded boundary
$\partial\Omega$ in $\R^d$ for $d\geq 2$.\\
First, we assume
\begin{assumption}\label{ass1}
$|\nabla V (x)|$ never vanishes in $\overline {\Omega}$.
\end{assumption}
Note that together with Assumption \ref{ass7} this implies that $V$
satisfies \eqref{ass3} and \eqref{ass4} in $\overline{\Omega}$.\\

Let $\partial\Omega_\perp$ denote the subset of $\partial\Omega$ where
$\nabla V$ is orthogonal to $\partial\Omega\,$:
\begin{equation}\label{defPaPerp}
 \partial\Omega_\perp = \{x\in\partial\Omega^\#: \nabla V(x) = (\nabla V(x) \cdot \vec \nu(x))\,\vec \nu(x) \}\,,
\end{equation}
where $\vec \nu(x)$ denotes the outward normal on $\partial\Omega $ at
$x$\,.\\
We now recall from \cite{AGH} the definition of the one-dimensional complex Airy
operators.  To this end we let $ \mathfrak D^\#$, for $\#\in\{D,N\}$, be defined in the following
manner
\begin{equation}
\label{eq:98}
  \begin{cases}
    \mathfrak D^\#=\{u\in H^2_{loc}(\overline{\R_+}) \,| \, u(0)=0 \} & \#=D \\
    \mathfrak D^\# = \{u\in H^2_{loc}(\overline{\R_+}) \,| \, u^\prime(0)=0 \} & \#=N \,.
  \end{cases}
\end{equation}

Then, we define the operator
\begin{displaymath}
   \LL^\#(\jf ) = -\frac{d^2}{dx^2}+i\, \jf \,x \,,
\end{displaymath}
whose domain is given by
\begin{equation}
\label{eq:101}
  D(\LL^\# (\jf ))= H^2(\R^+)\cap L^2(\R^+;|x|^2dx)\cap \mathfrak D^\#\,,
\end{equation}
and set
\begin{equation} 
\label{eq:75}
  \lambda^\#(\jf ) = \inf \Re\sigma(\LL^\#(\jf )) \,.
\end{equation}
Next, let 
\begin{equation}
\label{eq:97} 
\Lambda_{m}^\#=\inf_{x\in\partial\Omega_\perp}\lambda^\#(|\nabla V(x)|) \,.
\end{equation}
In all cases we denote by $\Sg^\#$ the set 
\begin{equation}\label{eq:2}
  \Sg^\#:=\{x\in\partial\Omega_\perp \,: \,\lambda^\# (|\nabla V(x)|) = \Lambda_{m}^\#\,\}\,.
\end{equation}

When $\#\in\{D,N\}$ it can be verified by a dilation argument that,
when $\jf >0\,$,
\begin{equation}
   \lambda^\#(\jf) =\lambda^\#(1)\, \jf ^{2/3}\,.
\end{equation}
Hence 
\begin{equation}
\Lambda_{m}^\#= \lambda^\# (\jf_m)\,,\mbox{ with } \jf_m:= \inf_{x\in\partial\Omega_\perp}(|\nabla V(x)|)\,,
\end{equation}
and $\Sg^\#$ is actually independent of $\#$:
\begin{equation}\label{eq:2a}
  \Sg^\#=\Sg:=\{x\in\partial\Omega_\perp \,: \,|\nabla V(x)| = \jf_m\,\}\,.
\end{equation}

We next make the following additional assumption:
\begin{assumption}\label{nondeg2}
At each point $x$ of $\Sg^\#$, 
\begin{equation}
\label{eq:105}
 \alpha (x)=\text{\rm det} D^2V_\partial(x) \neq 0\,,
\end{equation}
where $V_\partial$ denotes the restriction of $V$ to $\partial\Omega
\,$, and $D^2V_\partial$ denotes its Hessian matrix.
\end{assumption}
It can be easily verified that \eqref{eq:105} implies that $\mathcal
S^\#$ is finite.  Equivalently we may write
\begin{subequations}
   \label{eq:95}
\begin{equation} 
\alpha(x) =\Pi_{i=1}^{d-1}\alpha_i(x) \neq0\,,
\end{equation}
where $\alpha_1,\ldots,\alpha_{d-1}$ are the eigenvalues of the
Hessian matrix $D^2V_\partial(x)$:
\begin{equation} 
  \{\alpha_i\}_{i=1}^{d-1} = \sigma(D^2V_\partial) \,,
\end{equation}
where each eigenvalue is counted according to its multiplicity. 
\end{subequations}

Our main result is
\begin{theorem}
\label{theorem2.7} 
Under Assumptions \ref{ass7}, \ref{ass1} and \ref{nondeg2}, we have
\begin{equation}\label{limSpect1}  
  \lim\limits_{h\to0}\,\frac{1}{h^{2/3}}\inf \bigl\{\Re\, \sigma(\mathcal A _h^D) \bigr\} =  \Lambda_{m}^{D}\,,
\qquad  \Lambda_{m}^{D} = \frac{|a_1|}{2}\jf_m ^{2/3}\,,
% \varliminf\limits_{h\to0}\frac{1}{h^{2/3}}\inf \bigl\{\Re\, \sigma(\mathcal A _h^D) \bigr\} \geq  \frac{|a_1|}{2}\jf_m ^{2/3}\,,
\end{equation}  
where $a_1<0$ is the rightmost zero of the Airy function $\Ai$.
Moreover, for every $\varepsilon>0\,$, there exist $h_\varepsilon>0$
and $C_\varepsilon>0$ such that
\begin{equation}\label{estRes1}  
 \forall h\in(0,h_\varepsilon),~~~
 \sup_{
\begin{subarray}{c} 
\gamma\leq \Lambda_{m}^{D}  \\[0.5ex]
%\gamma\leq \frac{|a_1|}{2}\, \jf_m^{2/3} \\[0.5ex]
     \nu \in\mathbb{R}
\end{subarray}
}
 \|(\mathcal A _h^D-(\gamma -\varepsilon)h^{2/3}-i\nu)^{-1}\|\leq\frac{C_\varepsilon}{h^{2/3}}\,.
% \forall h\in(0,h_\varepsilon),~~~
% \sup_{\tiny{\begin{array}{c}\gamma\leq \frac{|a_1|}{2}\,\jf_m ^{2/3}\,,\\ \nu\in\mathbb{R}\end{array}}}
% \|(\mathcal A _h^D-(\gamma -\varepsilon)h^{2/3}-i\nu)^{-1}\|\leq\frac{C_\varepsilon}{h^{2/3}}\,.
\end{equation}
\end{theorem}
In its first part, this result is essentially a reformulation of the
result stated by the first author in \cite{Alm}.  Note that the second
part provides, with the aid of the Gearhart-Pr\"uss theorem, an
effective bound (with respect to both $t$ and $h$) of the decay of the
associated semi-group as $t\ar +\infty\,$.  The theorem holds in
particular in the case $V(x)=x_1$ where $\Omega$ is the complementary
of a disc (and hence $S^T$ consists of two points).  Note that $\jf_m
=1$ in this case.

\begin{remark}
A similar result can be proved for the Neumann case where
\eqref{limSpect1} is replaced by
\begin{equation}\label{limSpect1N} 
 \lim\limits_{h\to0}\,\frac{1}{h^{2/3}}\inf \bigl\{\Re\, \sigma(\mathcal A _h^N) \bigr\} =  \Lambda_{m}^{N}\,,
\qquad  \Lambda_{m}^{N} = \frac{|a'_1|}{2}\, \jf_m ^{2/3}\,,
% \varliminf\limits_{h\to0}\frac{1}{h^{2/3}}\inf \bigl\{\Re\, \sigma(\mathcal A _h^N) \bigr\} \geq  \frac{|a'_1|}{2}\, \jf_m ^{2/3}\,,
\end{equation}
where $a'_1<0$ is the rightmost zero of $\Ai^\prime$, and
\eqref{estRes1} is replaced by
\begin{equation}\label{estRes1N}  
 \forall h\in(0,h_\varepsilon),~~~
 \sup_{
\begin{subarray}{c} 
\gamma\leq \Lambda_{m}^{N}  \\[0.5ex]
%\gamma\leq \frac{|a_1'|}{2}\, \jf_m^{2/3} \\[0.5ex]
     \nu \in\mathbb{R}
\end{subarray}
}
 \|(\mathcal A _h^N-(\gamma -\varepsilon)h^{2/3}-i\nu )^{-1}\|\leq\frac{C_\varepsilon}{h^{2/3}}\,.
\end{equation}
One can also treat the Robin case or the transmission case (see
\cite{AGH}).
\end{remark}

In the case of the Dirichlet problem, this theorem was obtained in
\cite[Theorem 1.1]{AlHen} for the interior problem and under the
stronger assumption that, at each point $x$ of $ \Sg^D$, the Hessian
of $ V_\partial:= V_{/\partial \Omega^\#}$ is positive definite if
$\partial_\nu V (x) < 0\,$ or negative definite if $\partial_\nu V (x)
> 0\,$, with $\partial_\nu V:=\nu\cdot \nabla V$.  This was extended
in \cite{AGH} to the interior problem without the sign condition of
the Hessian.  Here we prove this theorem for the exterior problem.

\section{Determination of the essential spectrum}
\label{s2}

\subsection{Weyl's theorem for non self-adjoint operators}
\label{sec:weyl}
For an operator which is closed but not self-adjoint, there are many
possible definitions for the essential spectrum.  We refer the reader
to the discussion in \cite{GuWe} or \cite{Sc} for some particular
examples.  In the present work, we adopt the following definition
\begin{definition}
Let $A$ be a closed operator. We will say that $\lambda\in
\sigma_{ess} (A)$ if one of the following conditions is not satisfied:
\begin{enumerate}
\item The multiplicity $\alpha(A-\lambda)$ of $\lambda$ is finite.
\item The range $R(A-\lambda)$ of $(A-\lambda)$ is closed.
\item The codimension $\beta(A-\lambda)$ of $R(A-\lambda)$ is finite.
\item $\lambda$ is an isolated point of the spectrum.
\end{enumerate}
\end{definition}
For bounded selfadjoint operators $A$ and $B$, Weyl's
theorem states that if\break  $A-B=W$ is a compact operator, then
$\sigma_{ess}(A)=\sigma_{ess}(B)$.\\
Once the requirement for self-adjointness
is dropped, a similar result can be obtained, though not without
difficulties (see \cite{GuWe}).  We thus recall the following
theorem from \cite[corollary 2.2]{Sc} (see also \cite[corollary
11.2.3 ]{Dav}). 
\begin{theorem} 
\label{thm:weyl}
Let  $A$ be a bounded operator and $B=A+W$.  If $W$ is compact, then
\begin{equation*}
  \sigma_{ess} (B)=\sigma_{ess} (A)\,.
\end{equation*}
\end{theorem}

In the present contribution we obtain the essential spectrum of
$(\mathcal A_h^\# +1)^{-1}$, which is clearly a bounded operator in
view of the accretiveness of $\A_h^\#$. We follow arguments
disseminated in \cite{PuRo} (see also \cite{Sam}), that are
rather standard in the self-adjoint case.  The idea is to compare
two bounded operators in $\mathcal L (L^2(\mathbb R^d))$. The proof is
divided into two steps.

\subsection{The pure Bloch-Torrey case in $\R^d$}
\label{s4}

We consider the case where $V(x)=V_0(x)$ and $V_0$ is given by 
\begin{equation*}
V_0(x):= \jf \,x_1
\end{equation*}
with $\jf \neq 0$ (assuming $h=1$) and will apply the result of
Subsection \ref{sec:weyl}.  The first operator is, in $\mathcal L (L^2(\mathbb
R^d))$
\begin{equation*}
A= (\A_0 + 1)^{-1}\,,
\end{equation*}
where
\begin{displaymath}
  \A_0=-\Delta+i\,V_0\,. 
\end{displaymath}
Because $d\geq 2$,  $A$ is not compact but nevertheless
we have 
\begin{lemma}
\begin{equation}
\label{assertion}
\sigma (A) =
\{0\}\,.
\end{equation}
\end{lemma}
{\bf Proof}\\
To prove that $\sigma (A) \subseteq\{0\}$\, we use the property that
\begin{equation*}
  \lambda \in \sigma (-\Delta + i \,V_0(x) +1) \mbox{ iff } \lambda^{-1} \in \sigma (A)\setminus \{0\}
\end{equation*}
and similarly
\begin{equation*}
  \lambda \in \sigma_{ess} (-\Delta + i \,V_0(x) +1) \mbox{ iff } \lambda^{-1} \in \sigma_{ess} (A)\setminus \{0\}\,.
\end{equation*}
However, it has been established in \cite{Alm,AGH} that the spectrum of $(-\Delta + i\, V_0
+1)$ is empty and hence $\sigma (A) \subseteq\{0\}$. \\
To prove that $0\in\sigma(\A)$ we
consider first the one-dimensional operator
\begin{displaymath}
  \LL=-\partial^2_{x_1}+i\, V_0(x_1) \,,
\end{displaymath}
defined on $D(\LL)=H^2(\R)\cap L^2(\R;x^2dx)$.\\
Since $(\LL+1)^{-1}$ is
compact, it follows that there exists $\{f_k\}_{k=1}^{+\infty} \subset L^2(\R)$ such that
$\|f_k\|_2=1$ and $\phi_k\overset{def}{=}(\LL+1)^{-1}f_k\to0$.\\ Let
$\psi\in C_0^\infty(\R^{d-1})$ satisfy $\|\psi\|_2=1$ and  further
$g_k(x)=f_k(x_1)\psi(x^\prime)- \phi_k(x_1)\Delta_{x^\prime}\psi$.\\ It can be easily
verified that $$Ag_k=\phi_k\psi\to0\,,\mbox{ with } \|g_k\|_2\to1\,.
$$
 Hence,
$0\in\sigma(A)$ and the lemma is proved. 
\qed 

For a given regular set $K$ with non empty interior, 
consider in $L^2(\mathbb R^d)$ (which is identified with $
L^2(\dot{K}) \oplus L^2(\Omega)$ where  $\dot{K}$  is the interior of $K$)    
the operator
\begin{equation*}
B := 0 \oplus (\mathcal A^N_{\Omega,V_0} +1)^{-1}\,.
\end{equation*}
Again we have, 
\begin{equation*}
 \lambda \in \sigma (\mathcal A^N_{\Omega,V_0}+1) \mbox{ iff } \lambda^{-1}  \in \sigma (B) \setminus\{0\}\,,
\end{equation*}
and similarly
\begin{equation*}
 \lambda \in \sigma_{ess} (\mathcal A^N_{\Omega,V_0}+1) \mbox{ iff } \lambda^{-1}  \in \sigma_{ess} (B) \setminus\{0\}\,,
\end{equation*}
Hence it remains to prove:
\begin{proposition}
\label{prop:essential}
\begin{equation*}
\sigma_{ess} (B)  = \sigma_{ess} (A)\,.
\end{equation*}
\end{proposition}
By Weyl's theorem, it is enough to prove.
\begin{proposition}
\label{prop:compact}
$A-B$ is a compact operator.
\end{proposition}
\begin{proof}
 We follow the proof in \cite[p. 578-579 ]{PuRo} (with suitable
changes due to the non self-adjointness of $A$ and $B$).  To this 
end, we introduce the intermediate operator
\begin{equation}
C:=   (\mathcal A^D_{\dot{K},V_0} +1)^{-1} \oplus (\mathcal A^N_{\Omega,V_0} +1)^{-1} \,,
\end{equation}
where $\mathcal A^D_{\dot{K},V_0} $ is the Dirichlet realization of
$(-\Delta + i V (x))$ in $\dot{K}$.\\
It is clear that $ C-B$ is compact, hence it is now enough to obtain
the compactness of the operator $C-A$.\\
  For $f, g\in L^2(\mathbb
R^d)$, let
\begin{equation*}
u = A f\,,\quad  v= C ^* g\,.
\end{equation*}
We then define
\begin{equation*}
u_+ = u_{/\Omega}\,,\quad u_- = u_{/\dot {K}}\,,\quad v_+=v_{/\Omega}\,,\quad v_- =v_{/\dot {K}}.
\end{equation*}
Note that
\begin{equation*}
v_- = ((\mathcal A^D_{\dot{K},V_0} +1)^*)^{-1} g_-= (\mathcal A^D_{\dot{K},-V_0} +1)^{-1} g_-\,,
\end{equation*}
and
\begin{equation*}
v_+ = ((\mathcal A^N_{\Omega,V_0} +1)^*)^{-1} g_+= (\mathcal A^N_{\Omega,-V_0} +1)^{-1} g_+\,.
\end{equation*}

We now write
\begin{equation*}
\begin{array}{ll}
\langle (A-C)f,g\rangle &  = \langle u ,  \left( (\mathcal A^D_{\dot{K},V_0} +1)  \oplus (\mathcal A^N_{\Omega,V_0} +1)\right)^* v 
\rangle - \langle  (\mathcal A_0  +1)  u\,,\, v \rangle\\
 &= \langle u_+, (-\Delta) v_+\rangle_{L^2(\Omega)} + \langle u_-, (-\Delta) v_-\rangle_{L^2(\dot{K})} \\
  & \quad \quad - \langle (-\Delta) u_+, v_+\rangle_{L^2( \Omega)} - \langle
(-\Delta )u_-, v_-\rangle_{L^2(\dot{K})}\,.
\end{array}
\end{equation*}
As $v_-$ satisfies a Dirichlet condition on
$\Gamma= \partial K=\partial \Omega$ and $v_+$ satisfies a Neumann
condition we obtain via integration by parts 
\begin{equation}
\langle (A-C) f,g\rangle = \int_\Gamma \left( u_- \,  \overline{\partial_\nu v_-} -\partial _\nu u_+ \,  \overline{ v_+}\right)  \, ds\,.
\end{equation}
To complete the proof we notice that by Sobolev embedding and the
boundedness of the trace operators  we have for
some compact $\tilde K$ such that $K \subset \dot{\tilde{K}}$ and some constants $C_{\tilde K}$, $C'_{\tilde K}$
\begin{displaymath}
\begin{array}{ll}
  |\langle (A-C) f,g\rangle|&  \leq C_{\tilde K} \left(\|u_+\|_{H^{3/2}(\dot{\tilde{K}}\setminus K)} + \|u_-\|_{H^{3/2}( \dot{K})} \right)\, \left(\|
  v_+\|_{H^2(\dot{\tilde{K}}\setminus K)} +\|v_-\|_{H^2(\dot{K})}\right)\\
   &  \leq C'_{\tilde K}  \, \|u\|_{H^{3/2}(\dot{\tilde{K}})}  \, \| g\|_2 \,.
   \end{array}
\end{displaymath}
Hence, 
\begin{equation}
\label{eq:1} 
  \|(A-C)f\|_2 \leq C'_{\tilde K} \|Af\|_{H^{3/2}(\dot{\tilde{K}})}\,.
\end{equation}
Let $\{ f_k\}_{k=1}^\infty\subset L^2(\R^d)$ satisfy $\|f_k\|\leq1$ for all
$k\in\N$. By the boundedness of $A$ in $\mathcal L (L^2(\mathbb R^d),
H^2(\mathbb R^d))$ the sequence $ \|Af_k\|$ is bounded in
$H^2(\dot{\tilde K})$. By Rellich's theorem, the injection of
$H^2(\dot{\tilde K})$ in $H^\frac 32(\dot{\tilde K})$ is compact.
Hence there exists a subsequence $\{ f_{k_m}\}_{m=1}^\infty$ such that
$\{Af_{k_m}\}_{m=1}^\infty$ is a Cauchy sequence in
$H^{3/2}(\dot{\tilde{K}})$. By \eqref{eq:1}
$\{(A-C)f_{k_m}\}_{m=1}^\infty$ is a Cauchy sequence in $L^2(\R^d)$ and
hence convergent. This completes the proof of Proposition
\ref{prop:compact} and hence also of Proposition \ref{prop:essential}.
\end{proof}

\begin{proof}[Proof of Theorem \ref{theorem2.5}]
  To prove Theorem \ref{theorem2.5} under Assumption \ref{ass5} for
  the case $\tilde{V}\not\equiv0$ we write, for some $\lambda\in\C$ with $\Re\lambda<0$
  \begin{displaymath}
    (-\Delta+i(V_0+\tilde{V})-\lambda)^{-1}=  (-\Delta+iV_0-\lambda)^{-1}[1-i\tilde{V} (-\Delta+i(V_0+\tilde{V})-\lambda)^{-1}]\,.
  \end{displaymath}
Since both  $(-\Delta+i(V_0+\tilde{V})-\lambda)^{-1}:L^2(\Omega)\to H^2(\Omega)$ and
$(-\Delta+iV_0-\lambda)^{-1}:L^2(\Omega)\to H^2(\Omega)$ are bounded, and since
$\tilde{V}:H^2(\Omega)\to L^2(\Omega)$ is compact (as a matter of fact
$\tilde{V}:H^1(\Omega)\to L^2(\Omega)$ is compact as well), it follows by Theorem~\ref{thm:weyl} that the essential spectrum of
$(-\Delta+i(V_0+\tilde{V})-\lambda)^{-1}$ is an empty set. This completes the
proof of Theorem \ref{theorem2.5} for the case $\#=N$.

To prove Theorem \ref{theorem2.5} for the case $\#=D$ we may follow
the same procedure as in Proposition \ref{prop:compact} to obtain a
slightly different compact trace operator, or apply the
following simple argument: Let $R$ be sufficiently large so that
$K\subset B(0,R)$. Let $\A^D_{B(0,R)\setminus K} $ denote the Dirichlet realization
of $\A$ in $B(0,R)\setminus K$. By Proposition~\ref{prop:compact}, the operators $(\mathcal
A^N_{\R^d\setminus B(0,R)} +1)^{-1}\oplus (A^D_{B(0,R)\setminus K} +1)^{-1}$ and
$(A^D_{\Omega} +1)^{-1}$, both in $\mathcal L (L^2(\Omega))$, differ by a
compact operator. Hence, as $$\sigma_{ess}\big((\mathcal
A^N_{\R^d\setminus B(0,R)} +1)^{-1}\oplus (A^D_{B(0,R)\setminus K} +1)^{-1}\big)=\emptyset\,,$$ we
obtain that $\sigma_{ess}\big((A^D_{\Omega} +1)^{-1}\big)=\emptyset$ as well. 
\end{proof}

\begin{remark}
\label{rem:empty}
  An essentially identical proof permits the comparison of the
  essential spectrum of the two exterior problems $(-\Delta +i V_1)^\sharp$ in
  $\Omega_1= \mathbb R^d\setminus K_1$ (with $\#\in\{D,N\}$) and $(-\Delta +i V_2)^\flat$
  in $\Omega_2=\mathbb R^d\setminus K_2$ (with $\flat \in \{D,N\}$) under the
  condition that $V_1=V_2$ outside a large open ball containing $K_1$
  and $K_2$.
  \end{remark}

\begin{proof}[Proof of Theorem \ref{theorem2.3}]
  Since the proof relies on semi-classical analysis, we reintroduce
  the parameter $h$ (we no longer assume $h=1$).  Under Assumption
  \ref{ass7} there exists $R>0$ such that $K\subset B(0,R)$ and a potential
  $V_0$ satisfying Assumptions \eqref{ass3} and \eqref{ass4} in $\R^d$
  and such that $V\equiv V_0$ in $\R^d\setminus B(0,R)$.  By Remark
  \ref{rem:empty} we need only consider the case when $K=\emptyset$,
  with $V$ satisfying \eqref{ass3} and \eqref{ass4} in $\R^d$ .
  
We use the same framework as in \cite{Hen,AGH}. We cover $\mathbb R^d$
by balls $B(a_j,h^\rho)$ of size $h^\rho$ ($\frac 13 < \rho <\frac 23$) and
consider an associated partition of unity $\chi_{j,h}$ such that
\begin{itemize}
\item 
$ \sum_{j\in \mathcal J_i(h)}\chi_{j,h}(x)^2 =1\,,$
\item ${\rm supp\,} \chi_{j,h} \subset B(
a_j(h), h^\rho)\,,$
\item For $|\alpha|\leq 2$,
$
\sum_j |\pa^\alpha\chi_{j,h}(x)|^2\leq C_\alpha \, h^{- 2 |\alpha|{\varrho}} \,.
$
\end{itemize}
$\Lambda$ being given, we construct the approximate resolvent
$(\mathcal A_h -z)$ (with $\Re z \leq \Lambda h^{\frac 23}$) by
\begin{equation*}
\mathcal R_h:= \sum_{j\in \mathcal J} \chi_{j,h} (\mathcal A_{j,h} -z)^{-1}\chi_{j,h}\,.
\end{equation*}
We then use the uniform estimate \cite{Hen}:
\begin{equation}\label{unifest}
\sup_{\Re z \leq \omega h^\frac 23} \| (\mathcal A_{j,h} -z)^{-1}\| \leq C_\omega
[\jf h]^{-\frac 23}\,,
\end{equation} 
where $\jf=|\nabla V(a_j)|$,  $C_\omega$ is independent of $j$, $h\in (0,h_0]$ and
\begin{equation}
\label{eq:6}
\mathcal A_{j,h} := -h^2 \Delta + i \, V_0(a_j) + i \,\nabla V_0(a_j)\cdot (x-a_j)
\end{equation}
is the linear approximation of $\mathcal A_h$ at the point $a_j\,$.\\
As in \cite{Hen,AGH}, we then get 
\begin{equation}
\label{eq:4}
\mathcal R_h \circ (\mathcal A_h -z) = I + \mathcal E (h)\,,
\end{equation}
where
\begin{displaymath}
  \mathcal E(h)  = \sum_{j\in\Jg}\chi_{j,h}\, i\,\big(V_0-V_0(a_j)
  -\nabla V_0(a_j)\cdot (x-a_j)\big)(\mathcal A_{j,h} -z)^{-1}\chi_{j,h}- h^2[\Delta,\chi_{j,h}] (\mathcal A_{j,h} -z)^{-1}\chi_{j,h}\,.
\end{displaymath}
The estimation of the second term in the sum can be done in precisely
the same manner as in \cite{Hen}. For the first term we have by
\eqref{ass3}
\begin{displaymath}
  \big\|\chi_{j,h} \big(V_0-V_0(a_j)
  -\nabla V_0(a_j)\cdot (x-a_j)\big)(\mathcal A_{j,h}
  -z)^{-1}\chi_{j,h}\big\|\leq C_\omega h^{2\rho-2/3}\,.
\end{displaymath}
By the above and \cite{Hen}
\begin{equation}
\label{eq:3}
\| \mathcal E(h) \|_{\mathcal L (L^2(\mathbb R^d))} = \mathcal O
(h^{2-2\rho -\frac 23}) + \mathcal O (h^{2\rho -\frac 23})\,. 
\end{equation}
To obtain \eqref{eq:3}, use has been made of \eqref{ass3},
\eqref{ass4} (of Assumption \ref{ass1}) which permit the use of
\eqref{unifest}. The bound on $|D^2V_0|/|\nabla V_0|^{2/3}$ is necessary in
order to estimate the error in the linear approximation of $V$ in the
ball $B(a_j, h^\rho)$.  Note that the cardinality of $\mathcal J_{i}
(h)$ is now infinite, but it has been established in \cite{AGH} that
the cardinality of the balls $B(a_k, 2h^\rho)$ intersecting a given
$B(a_j, h^\rho)$ is uniformly bounded in
$j$, $h$.\\
By \eqref{eq:3} $I+\Eg(h)$ is invertible for sufficiently small
$h$. Hence, by \eqref{eq:4} we have that
\begin{displaymath}
  \sup_{\Re z \leq \Lambda h^{\frac 23}}\|(\mathcal A_h -z)^{-1}\|\leq C\sup_{\Re
    z \leq \Lambda h^{\frac 23}}\|\mathcal R_h \|\leq \frac{C_\Lambda}{h^{2/3}}\,,
\end{displaymath}
where $c_0$ is the lower bound on $|\nabla V_0|$ given in \eqref{ass4}. 
We may now conclude that for any $\Lambda$, the spectrum (including the
essential spectrum) of $\mathcal A_h =
-h^2 \Delta + i  V$ in $\mathbb R^d$ is contained in $\{ z\in \mathbb C
\,|\, \Re z \geq \Lambda [c_0h]^\frac 23\}$ for $h$ small enough. 
\end{proof}

\section{The left margin of the spectrum}

This section is devoted to the proof of Theorem \ref{theorem2.7}. As
the proof is very similar to the proof in a bounded domain \cite{AGH},
and therefore we bring only its main ingredients.

\subsection{Lower bound}

By lower bound, we mean
\begin{equation}
  \label{eq:5}
 \varliminf\limits_{h\to0}\frac{1}{h^{2/3}}\inf \bigl\{\Re\, \sigma(\mathcal
 A _h^D) \bigr\} \geq   \Lambda_{m}^{\#}\,,
\end{equation}
where $ \Lambda_{m}^{D}$ is given in \eqref{limSpect1} and $ \Lambda_{m}^{N}$ in
\eqref{limSpect1N}. 
.

We keep the notation of \cite[Section~6 ]{AGH}.  For some $1/3<\varrho<2/3$ and for every $h\in(0,h_0]$,
we choose two sets of indices $\Jg_{i}(h)\,$, $ \Jg_{\partial}(h) \,$,
and a set of points
\begin{subequations}   \label{eq:25}
\begin{equation}
\big\{a_j(h)\in \Omega : j\in \mathcal J_i(h)\big\}\cup\big\{b_k(h)\in\pa {\Omega} : k\in \mathcal J_\partial (h)\big\}\,,
\end{equation}
such that $B(a_j(h),h^\varrho)\subset\Omega$\,,
\begin{equation}
    \bar\Omega \subset\bigcup_{j\in \Jg_{i}(h)}B(a_j(h),h^{\varrho})~\cup\bigcup_{k\in \Jg_{\partial }(h)}B(b_k(h),h^{\varrho})\,,
\end{equation}
and such that the closed balls $\bar B(a_j(h),h^{\varrho}/2)\,$, $\bar
B(b_k(h),h^{\varrho}/2)$ are all disjoint. \\

Now we construct in $\mathbb R^d$ two families of functions 
\begin{equation}
 (\chi_{j,h})_{j\in \mathcal J_i(h)} \mbox{ and } (\zeta_{j,h})_{j\in \mathcal J_\partial (h)}\,,
\end{equation}
and a function $\chi_{R,h}$ such that, for every $x\in\bar\Omega\,$,
\begin{equation}
\sum_{j\in \mathcal J_i(h)}\chi_{j,h}(x)^2+\sum_{k\in \mathcal J_\partial (h)}\zeta_{k,h}(x)^2=1\,,
\end{equation}
\end{subequations}
and such that 
\begin{itemize}
\item $\Supp \chi_{j,h}\subset B(a_j(h),h^{\varrho})$ for
$j\in \mathcal J_i(h)$,
\item
 $\Supp\zeta_{j,h}\subset
B(b_j(h),h^{\varrho})$ for $j\in \Jg_\partial\,$,
\item
$\chi_{j,h}\equiv 1$ (respectively $\zeta_{j,h}\equiv1$) on $\bar
B(a_j(h),h^{\varrho}/2)$ (respectively $\bar
B(b_j(h),h^{\varrho}/2)$)\,.
\end{itemize}

To verify that the approximate resolvent constructed in the sequel
satisfies the boundary conditions on $\partial\Omega $, we require in
addition that
\begin{equation}
 \label{eq:90}
\frac{\partial\zeta_{k,h}}{\partial\nu}\Big|_{\partial\Omega}=0 
\end{equation}
for $\#=N$.\\
Note that, for all $\alpha\in\mathbb{N}^n\,$, we can assume that there
exist positive $h_0$ and $C_\alpha$, such that, $\forall h \in
(0,h_0]$, $ \forall x\in \overline{\Omega}$,
\begin{equation}
\label{supPaCutoff} 
|\pa^\alpha \chi_{R,h}|^2+ \sum_j |\pa^\alpha\chi_{j,h}(x)|^2\leq C_\alpha \, h^{- 2 |\alpha|{\varrho}} ~~~\textrm{ and }~~~
\sum_j |\pa^\alpha\zeta_{j,h}(x) |^2 \leq C_\alpha  \, h^{-2 |\alpha|{\varrho}}\,.
\end{equation}

We now define the approximate resolvent as in \cite{AGH}
\begin{equation}
\label{eq:9}
  \mathcal R_h=\sum_{j\in \mathcal \Jg_i(h)}\chi_{j,h}(\A_{j,h}-\lambda)^{-1}\chi_{j,h} 
+  \sum_{j\in \mathcal \Jg_\partial (h)}\eta_{j,h}  R_{j,h}\eta_{j,h}\,,
\end{equation}
where $R_{j,h}$ is given by \cite[Eq. (6.14)]{AGH}, and
$\eta_{j,h}={\mathbf 1}_\Omega\zeta_{j,h}$.  As in \eqref{eq:4} we write 
\begin{equation}
\label{eq:7}
\mathcal R_h \circ (\mathcal A_h -z) = I + \mathcal E (h)\,,
\end{equation}
where
\begin{equation}
\label{eq:8}
\begin{array}{ll}
  \mathcal E(h)  & = \sum_{j\in\Jg}\chi_{j,h}(\A_h-\A_{j,h})(\mathcal A_{j,h} -z)^{-1}\chi_{j,h}\\ & \quad -
  h^2[\Delta,\chi_{j,h}] (\mathcal A_{j,h} -z)^{-1}\chi_{j,h} \\ & \quad  +\sum_{j\in
    \mathcal \Jg_\partial (h)}(\A_h-z)\eta_{j,h}  R_{j,h}\eta_{j,h}\,. 
    \end{array}
\end{equation}
The estimate of the first sum can be now made in  the same
manner as in the proof of Theorem \ref{theorem2.3}, whereas 
control of the second sum can be achieved as in \cite{AGH}. We may
thus conclude that for any $\epsilon>0$ there exists $C_\epsilon>0$ such that for
sufficiently small $h$
\begin{displaymath}
  \sup_{\Re z\leq h^{2/3}(\Lambda_m^\#-\epsilon)}\| \mathcal E(h) \|\leq C\, (h^{2-2\rho -\frac 23} + h^{2\rho -\frac 23})\,. 
\end{displaymath}
Since for sufficiently small $h$ $I+\Eg$ becomes invertible, we can
now use \eqref{eq:7} to conclude that  for any $\epsilon>0$ there exists $C_\epsilon>0$ such that for
sufficiently small $h$
\begin{displaymath}
   \sup_{\Re z\leq h^{2/3}(\Lambda_m^\#-\epsilon)}\|(\A_h-\lambda)^{-1}\|\leq
   \frac{C_\epsilon}{h^{2/3}} \,.
\end{displaymath}
This completes the proof of \eqref{eq:5}.  

\subsection{The proof of upper bounds}

To prove that 
\begin{equation*}
  \varlimsup\limits_{h\to0}\frac{1}{h^{2/3}}\inf \bigl\{\Re\, \sigma(\mathcal
 A _h^\#) \bigr\} \leq    \Lambda_{m}^{\#}\,,
\end{equation*}
we use the same procedure presented in \cite[Section 7]{AGH}. The only
thing we care to mention is that to estimate the contribution of the
interior of $\Omega$ (i.e. the first sum in \eqref{eq:9} and \eqref{eq:8})
we use the same approach as in the proof of Theorem \ref{theorem2.3}. The
rest of the proof, being precisely the same as in  \cite[Section
7]{AGH} is skipped.

\section{Numerical illustration}

In this section, we provide a numerical evidence for the existence of
a discrete spectrum of the Bloch-Torrey operator $\mathcal A_h^N =
-h^2 \Delta + ix_1$ in the case of the exterior of the unit disk:
$\Omega_\infty = \{ x\in\R^2~:~|x| > 1 \}$.  In contrast to the
remaining part of this note, this section relies on numerics and does
not pretend for a mathematical rigor: it only serves for illustration
purposes.

Since a numerical construction of the operator $\mathcal A_h^N$ is not
easily accessible for an unbounded domain, we consider the operator
$\mathcal A_{h,R}^N = - h^2 \Delta + ix_1$ in a circular annulus
$\Omega_R = \{ x\in\R^2~:~1 < |x|< R\}$ with two radii $1$ and $R$.
As $R\to + \infty\,$, the bounded domain $\Omega_R$ approaches to the
exterior of the disk $\Omega_\infty$.  We set Neumann boundary
condition at the inner circle and the Dirichlet boundary condition at
the outer circle.  Given that $\Omega_R$ is a bounded domain, the
operator $\mathcal A_{h,R}^N$ has a discrete spectrum (as $ix_1$ is a
bounded perturbation of the Laplace operator).  The operator $\mathcal
A_{h,R}$ can be represented via projections onto the Laplacian
eigenbasis by an infinite-dimensional matrix $-h^2 \Lambda +
i{\mathcal B}$, where the diagonal matrix $\Lambda$ is formed by
Laplacian eigenvalues and the elements of the matrix ${\mathcal B}$
are the projections of $x_1$ onto two Laplacian eigenfunctions (see
\cite{Grebenkov07,Grebenkov08,Grebenkov10,GH} for details).  In
practice, the matrix $-h^2 \Lambda + i{\mathcal B}$ is truncated and
then numerically diagonalized, yielding a well-controlled
approximation of eigenvalues of the operator $\mathcal A_{h,R}^N$, for
fixed $h$ and $R$.  For convenience, the eigenvalues are ordered
according their increasing real part.

As shown in \cite{GH}, for small enough $h$, the quasimodes of the
operator $\mathcal A_{h,R}^N$ are localized near the boundary of the
annulus, i.e., near two circles.  The quasimodes that are localized
near the inner circle are almost independent of the location of the
outer circle. Since the spectrum of the operator
$\mathcal A_h^N$ in the limiting (unbounded) domain $\Omega_\infty$ is
discrete, some eigenvalues of $\mathcal A_{h,R}^N$ are expected to
converge to that of $\mathcal A_h^N$ as $R$ increases.

Table \ref{tab:eigenvalues} shows several eigenvalues of the operator
$\mathcal A_{h,R}^N$ as the outer radius $R$ grows.  The symmetry of
the domain implies that if $\lambda$ is an eigenvalue, then the complex
conjugate $\bar{\lambda}$ is also an eigenvalue.  For this reason, we only
present the eigenvalues with odd indices with positive imaginary part.
One can see that the eigenvalues $\lambda_1$, $\lambda_3$ and $\lambda_7$ are
almost independent of $R$.  These eigenvalues correspond to the
eigenmodes localized near the inner circle.  We interpret this
behavior as the convergence of the eigenvalues to that of the operator
$\mathcal A_h^N$ for the limiting (unbounded) domain $\Omega_\infty$.  In
contrast, the imaginary part of the eigenvalues $\lambda_5$ and $\lambda_9$ grows
almost linearly with $R$, as expected from the asymptotic behavior
reported in \cite{GH}.  These eigenvalues correspond to the eigenmodes
localized near the outer circle and thus diverge as the outer circle
tends to infinity ($R \to +\infty$).  These numerical results illustrate the
expected behavior of the spectrum.  To illustrate the quality of the
numerical computation, we also present in Table \ref{tab:eigenvalues}
the approximate eigenvalues based on their asymptotics derived in
\cite{GH}:
\begin{equation}  \label{eq:lambda_app}
\begin{split}
\lambda^{N(n,k)}_{\rm app} & = i + h^{2/3} |a_n'| e^{\pi i/3} + h(2k-1) \frac{e^{-\pi i/4}}{\sqrt{2}} + h^{4/3} \frac{e^{\pi i/6}}{2|a'_n|} + O(h^{5/3}) , \\
\lambda^{D(n,k)}_{\rm app} & = i R + h^{2/3} |a_n| e^{-\pi i/3} + h(2k-1) \frac{e^{-\pi i/4}}{\sqrt{2R}} +  O(h^{5/3}) ,\\
\end{split}
\end{equation}
where $a_n$ and $a'_n$ are the zeros of the Airy function and its
derivative, respectively.  Note that
$\lambda^{N(n,k)}_{\rm app}$ corresponds to the inner circle of radius
$1$ where Neumann boundary condition is prescribed, whereas
$\lambda^{D(n,k)}_{\rm app}$ corresponds to the outer circle of radius
$R$ where we impose a Dirichlet boundary condition.  These approximate
eigenvalues (truncated at $O(h^{5/3})$) show an excellent agreement
with the numerically computed eigenvalues of the operator
$\A^N_{h,R}$.  This agreement confirms the accuracy of both the
numerical procedure and the asymptotic formulas
(\ref{eq:lambda_app}). 

%Some other numerical illustrations were provided in \cite{GH}.

\begin{table}
\begin{center}
\begin{tabular}{|c|c|c|c|c|}  \hline
$\lambda_n~\backslash~ R$  &    1.5     &      2       &       3          \\  \hline
%\multirow{2}{3mm}{1} 
$\lambda_1$ & $0.0250 + 1.0318i$ & $0.0250 + 1.0317i$ & $0.0251 + 1.0315i$ \\ %    N
\rowcolor[gray]{0.9}
$\lambda^{N(1,1)}_{\rm app}$ & $0.0251 + 1.0317i$ & $0.0251 + 1.0317i$ & $0.0251 + 1.0317i$ \\  \hline
%
%\multirow{2}{3mm}{3} 
$\lambda_3$ & $0.0409 + 1.0160i$ & $0.0409 + 1.0160i$ & $0.0410 + 1.0158i$ \\ %    N 
\rowcolor[gray]{0.9}
$\lambda^{N(1,3)}_{\rm app}$ & $0.0411 + 1.0157i$ & $0.0411 + 1.0157i$ & $0.0411 + 1.0157i$  \\ \hline 
%
%\multirow{2}{3mm}{5} 
$\lambda_5$ & $0.0501 + 1.4157i$ & $0.0497 + 1.9162i$ & $0.0498 + 2.9161i$ \\ %    D
\rowcolor[gray]{0.9}
$\lambda^{D(1,1)}_{\rm app}$ & $0.0500 + 1.4157i$ & $0.0496 + 1.9162i$ & $0.0491 + 2.9167i$  \\ \hline 
%
%\multirow{2}{3mm}{7} 
$\lambda_7$ & $0.0567 + 1.0003i$ & $0.0567 + 1.0003i$ & $0.0560 + 1.0000i$  \\ %    N
\rowcolor[gray]{0.9}
$\lambda^{N(1,5)}_{\rm app}$ & $0.0571 + 0.9997i$ & $0.0571 + 0.9997i$ & $0.0571 + 0.9997i$  \\ \hline  
%
%\multirow{2}{3mm}{9} 
$\lambda_9$ & $0.0635 + 1.4026i$ & $0.0612 + 1.9048i$ & $0.0593 + 2.9065i$  \\ %    D
\rowcolor[gray]{0.9}
$\lambda^{D(1,3)}_{\rm app}$ & $0.0631 + 1.4027i$ & $0.0609 + 1.9049i$ & $0.0583 + 2.9075i$  \\ \hline  
%
%& Re & 0.0250 & 0.0250 & 0.0251 & 0.0246 \\ %    N
%& Im & 1.0318 & 1.0317 & 1.0315 & 1.0304 \\  \hline
%%
%\multirow{2}{3mm}{3} 
%& Re & 0.0409 & 0.0409 & 0.0410 & 0.0400 \\ %    N 
%& Im & 1.0160 & 1.0160 & 1.0158 & 1.0150 \\ \hline
%%
%\multirow{2}{3mm}{5} 
%& Re & 0.0501 & 0.0497 & 0.0498 & 0.0507 \\ %    D
%& Im & 1.4157 & 1.9162 & 2.9161 & 3.9241 \\ \hline
%%
%\multirow{2}{3mm}{7} 
%& Re & 0.0567 & 0.0567 & 0.0560 & 0.0534 \\ %    N
%& Im & 1.0003 & 1.0003 & 1.0000 & 3.4976 ~??? \\ \hline 
%%
%\multirow{2}{3mm}{9} 
%& Re & 0.0635 & 0.0612 & 0.0593 & 0.0534 \\ %    D
%& Im & 1.4026 & 1.9048 & 2.9065 & 3.4341 \\ \hline 
\end{tabular}
\end{center}
\caption{
Several eigenvalues of the operator $\mathcal A_{h,R}^N$ in the
circular annulus $\Omega_R = \{x\in\R^2~:~ 1 < |x|< R\}$ computed
numerically by diagonalizing the truncated matrix representation $-h^2
\Lambda + i{\mathcal B}$, for $h = 0.008$ and $R = 1.5, 2, 3\,$.  For
comparison, gray shadowed lines show the approximate eigenvalues from
Eqs. (\ref{eq:lambda_app}). }
\label{tab:eigenvalues}
\end{table}

%1  2.500557e-002  2.500119e-002  2.505210e-002  2.461641e-002     N
%1  1.031750e+000  1.031733e+000  1.031517e+000  1.030382e+000  
%3  4.087983e-002  4.087620e-002  4.094679e-002  3.995634e-002     N 
%3  1.016028e+000  1.016011e+000  1.015781e+000  1.015024e+000  
%5  5.011705e-002  4.965442e-002  4.975151e-002  5.073628e-002  
%5  1.415717e+000  1.916164e+000  2.916145e+000  3.924111e+000     D
%7  5.669231e-002  5.667244e-002  5.604929e-002  5.339809e-002     N
%7  1.000296e+000  1.000296e+000  9.999611e-001  3.497594e+000 ?
%9  6.354863e-002  6.120614e-002  5.926983e-002  5.340401e-002  
%9  1.402557e+000  1.904786e+000  2.906539e+000  3.434078e+000     D

\section{Conclusion}\label{s7}

While we have confined the discussion in this work to Dirichlet and
Neumann boundary conditions for simplicity,
we could have also treated the Robin case or the transmission case (see
\cite{AGH}) with $\Omega^+ = \mathbb R^d \setminus
\overline{\Omega^{^-}}$.  Note that we do not assume that $K$ is
connected.  In the case of the Dirichlet problem, the main theorem was
obtained in \cite[Theorem 1.1]{AlHen} under the stronger assumption
that, at each point $x$ of $ \Sg^D$, the Hessian of $ V_\partial:=
V_{/\partial \Omega}$ is positive definite if $\partial_\nu V (x) <
0\,$ or negative definite if $\partial_\nu V (x) > 0\,$, with
$\partial_\nu V:=\nu\cdot \nabla V$.  This additional assumption
reflects some technical difficulties in the proof, that was overcome
in \cite{AGH} by using tensor products of semigroups, a point of view
that was missing in \cite{AlHen}. \\
This generalization allows us to obtain the asymptotics of the left
margin of $\sigma(\A_h^\#)$, for instance, when $V(x_1,x_2)=x_1$ and
$\Omega$ is the exterior of a disk, where the above assumption is not
satisfied. \\
For this particular potential, an extension to the case when $\Omega$
is unbounded is of significant interest in the physics literature
\cite{Grebenkov17}. \\
 
{\bf Acknowledgements:}\\
Y. almog was partially supported by the NSF Grant DMS-1613471. He also
wishes to thank the department of Mathematics at the Technion and his
host Itai Shafrir for supporting his visit there during the period when this
manuscript was written. B. Helffer would like to thank Vincent
Bruneau for useful discussions.

\end{document}